\documentclass{vldb}
\usepackage{graphicx}
\usepackage{balance}  

\usepackage{times}
\usepackage[noend,linesnumbered]{algorithm2e}
\usepackage{amsthm}
\usepackage{amsmath}
\usepackage{upgreek} 
\usepackage{enumitem}
\usepackage{wasysym}
\usepackage{stmaryrd}
\usepackage[us]{datetime}
\usepackage{xcolor} 
\usepackage[export]{adjustbox}
\usepackage[hang]{subfigure}
\usepackage{hyperref}
\usepackage{balance}

\usepackage{cite}

\newtheorem{lem}{Lemma}
\newtheorem{defn}{Definition}
\newtheorem{theo}{Theorem}
\newtheorem{examp}{Example}

\newcommand{\verysmallskip}{\vspace{0.5ex}}
\renewcommand{\paragraph}[1]{\verysmallskip\noindent\textbf{#1.}}

\newcommand{\pos}{\ensuremath{\rho}}

\newcommand{\avgwin}{\overline{|W|}}

\newcommand{\indextime}{\ensuremath{t_J}}
\newcommand{\stock}{S}
\newcommand{\eqoverlap}{overlap}

\newcommand{\myasygraphics}[2][]{\IfFileExists{#2.pdf}{\includegraphics[#1]{#2}}{\ \wlog{File: File #2.pdf not found: }}}

\def\HiLi{\leavevmode\rlap{\hbox to \hsize{\color{gray!30}\leaders\hrule height .8\baselineskip depth .5ex\hfill}}}

\vldbTitle{SWOOP: Top-k Similarity Joins over Set Streams}
\vldbAuthors{Willi Mann, Nikolaus Augsten, Christian S.\ Jensen}
\vldbDOI{https://doi.org/10.14778/xxxxxxx.xxxxxxx}
\vldbVolume{12}
\vldbNumber{xxx}
\vldbYear{2019}

\begin{document}


\title{SWOOP: Top-k Similarity Joins over Set Streams}



%
%
%
%

\numberofauthors{3} 

\author{
%
%
\alignauthor
Willi Mann\\
       \affaddr{Celonis SE}\\
       \affaddr{Munich, Germany}\\
       \email{w.mann@celonis.com}
\alignauthor
Nikolaus Augsten\\
       \affaddr{University of Salzburg}\\
       \affaddr{Salzburg, Austria}\\
       \email{nikolaus.augsten@sbg.ac.at}
\alignauthor 
Christian S.\ Jensen\\
       \affaddr{Aalborg University}\\
       \affaddr{Aalborg, Denmark}\\
       \email{csj@cs.aau.dk}
}

\maketitle

\begin{abstract}

We provide efficient support for applications that aim to continuously find pairs of similar sets in rapid streams of sets, such as streams of tweets that consist of sets of words.
Using a sliding window model, the top-$k$ result changes as new sets enter the window and existing ones leave the window. Specifically, when a set arrives, it may form a new top-$k$ result pair with any set already in the window, and when a set leaves the window, all its pairings in the top-$k$ result must be replaced with other pairs. It is insufficient to maintain the $k$ most similar pairs since less similar pairs may become top-$k$ pairs. 

We propose SWOOP, a highly scalable stream join algorithm. Novel indexing techniques and sophisticated filters efficiently prune useless pairs as new sets enter the window. SWOOP incrementally maintains a provably minimal stock of similar pairs to update the top-$k$ result at any time. Empirical studies confirm that SWOOP is able to support stream rates that are orders of magnitude faster than the rates supported by existing approaches.
\end{abstract}

\section{Introduction}

\setlength{\emergencystretch}{3em}

The decreasing latency between the production of data, including humans and a broad range of sensors, and consumption of data renders streaming data increasingly prevalent. We consider streams where the elements of the streams are timestamped sets. Examples of such elements include tweets, email messages, or news articles that may be modeled as sets of words or $n$-grams; retail point-of-sale transactions represented as sets of goods; the clicks in user click-streams on a website; or social media content represented by the sets of users that liked or consumed that content. 

Such data streams may achieve very high frequencies. For example, Apple's Siri user base may issue billions of requests per month; that may be modeled as sets of words or other signatures. As another example, Twitter emits about half a billion tweets per day. To analyze such rapid data streams, new techniques must be developed that can keep up with high-rate streams, including their peak rates. As new data items arrive in a stream, they are queued and processed in FIFO order. When the processing cannot keep up with the stream rate, the queue grows and leads  to waiting times for all subsequent data items. Delays between an event and its visibility in the result are critical in situation when events require timely action, e.g., blocking a spamming email account~\cite{hariharan-uspatent-2014_detecting-bulk-email}.

We consider the problem of computing the top-$k$ join in rapid data streams of timestamped sets with a sliding window, i.e., we compute all pairs of sets that are among the top-$k$ most similar pairs in a time window of duration $w$.  As new data items arrive in the stream, the window moves, and the top-$k$ result must be updated. The top-$k$ join over streams may, for example, be used to recommend products based on recent point-of-sales transactions or click-stream data~\cite{montgomery2004modeling,Wang:2016:UCC:2858036.2858107}, to aggregate similar trending IPA requests to improve answer quality (e.g., by sharing successful interactions with users of similar requests), or to detect trends or to analyze information diffusion in streams of tweets~\cite{jung-icis-2018_twitter-information-diffusion}. 
%
%
%

The top-$k$ join with a sliding window is useful also for static data, where the window covers all data elements whose timestamp falls within the window. A set pair is in the join result if it is among the top-$k$ in any interval of duration $w$. For example, consider an ERP system in which users scan and upload documents and where near-duplicate documents should be detected (e.g., to avoid paying  a bill twice). Each document is represented by a set of words resulting from an OCR process.  Computing all pairs of near-duplicate documents in the entire database will typically lead to many irrelevant result pairs since documents of interest are uploaded within a small time frame, e.g., some weeks. Therefore, only pairs within a given time window should be considered.  


We model a \emph{stream} as a sequence of $(\mathit{set}, \mathit{timestamp})$ pairs with monotonically increasing timestamps.  Only set pairs that are covered by a sliding time window $W$ of duration $w$ are considered. As the window slides over the stream, newly arriving sets become part of the window, and sets expire as they get older than the window duration. The top-$k$ join result must be kept up-to-date when time passes and such changes occur. 
Maintaining the join result poses two main challenges. (1) \emph{Candidate generation:} New sets that enter the sliding window may form a pair with any of the existing sets in the window. (2) \emph{Result expiration:} When sets expire, all their pairings become invalid; expired pairs among the \mbox{top-$k$} must be removed, and replacements must be found to maintain a correct join result. We next discuss these challenges in detail.

\emph{Candidate generation:} A new set that enters the window may form a pair with any of the $|W|$ sets in sliding window $W$. In rapid streams, the sliding window may contain hundreds of thousands of sets, so computing the similarity between each new set and all sets in the window does not scale to fast stream rates. Well known similarity join techniques for static set collections rely on inverted list indices~\cite{RBYMRS07,boge12,maau14,wang-pvldb-2017,xiao-www-08} that store a posting list of candidate sets for each token (or for each signature~\cite{deng-pvldb-2015}). Many techniques used in static scenarios, where all sets are known up front, cannot be used for streams, e.g., we cannot order tokens by their frequency or process and index sets in non-decreasing size order. Further, an index for streams must remove expired sets, which is expensive in indexes for static data. Finally, core technologies like the prefix filter~\cite{SCVGRK06} that are leveraged in this context use a threshold, whereas our scenario has no threshold because a top-$k$ result is required.


A top-$k$ join algorithm over a static collection of sets is proposed by Xiao et al.~\cite{conf/icde/XiaoWLS09}. A fundamental assumption of this approach, which is leveraged for pruning and index construction, is that all sets are known up front. There is no obvious way to adapt the static top-$k$ join to our dynamic setting with frequent new and expiring sets. Reevaluating the static top-$k$ join each time the sliding window changes does not scale to frequent changes, as we show in our empirical evaluation. Note that an approximate algorithm that processes updates in batches may introduce a large error: (a) Each new set in the window can form $|W|$ pairs that are more similar than all  pairs in the previous window, therefore invalidating the previous top-$k$ result. (b) Relevant pairs may never appear together in a window when the window is moved in batches; and increasing the window duration to $w'>w$ such that both the old window and the batch are covered does not solve the problem. 
 
\emph{Result expiration.} As time passes, sets leave the sliding window and expire. When a set expires, all pairs in the top-$k$ result containing the expired set must be removed, and the invalidated pairs must be replaced by other pairs. Thus, it is insufficient to keep maintain the \mbox{top-$k$} pairs; rather, a \emph{stock} of other, less similar, valid pairs must be maintained. The total number of valid pairs is quadratic in the window size, so maintaining all such pairs is not efficient for large sliding windows or rapid streams. Only relevant pairs that may be required later to maintain a correct result should be stored. The state of the art solution is SCase~\cite{Shen2014}, which computes a so-called skyband to remove all irrelevant pairs.
However, the skyband for the stock must be recomputed from scratch for every new set in the stream. The stock stores $O(k\cdot|W|)$ pairs that must all be touched to recompute the skyband.  As a result, SCase does not scale to rapid streams, and new approaches are required.

We propose SWOOP for top-$k$ joins over streaming sets. SWOOP uses a novel candidate index to efficiently generate a small set of candidate pairs when new sets enter the sliding window. Each new set in the stream forms $O(k\cdot|W|)$ new pairs that may be relevant. The candidate index leverages a lower bound derived from the skyband to prune candidate pairs. The lower bound must be computed for each pair under consideration, and the lower bound changes with every new pair that is inserted into the stock. We propose a new technique that computes the skyband lower bound in logarithmic time, and stock updates do not incur any cost. Previous approaches require linear time to update the skyband lower bound~\cite{Shen2014}. The cost of updating the candidate index in response to new or expiring sets is independent of the index size.


To efficiently maintain the stock of relevant pairs, we propose a novel technique to incrementally update the skyband; this technique does not depend on set similarity and is applicable to general stream join frameworks~\cite{Shen2014}.  
We show experimentally that this incremental stock update maintains the skyband for streams at rates that are up to ten times faster than the rates processed by the state-of-the-art solution SCase~\cite{Shen2014}. When combined with the candidate index, we achieve speed-ups of up to three orders of magnitude compared to an SCase-based approach.


To characterize the similarity functions to which SWOOP is applicable, we define the concept of \emph{well-behaved} similarity function. All standard set similarity functions are well-behaved, including Overlap, Jaccard, Cosine, Dice, and Hamming~\cite{conf/icde/XiaoWLS09}.

Finally, we report on an extensive experimental study that offers insight into the efficiency of SWOOP compared to SCase~\cite{Shen2014}, static top-$k$ join~\cite{conf/icde/XiaoWLS09}, and a baseline. Most notably, we find that SWOOP scales much better with a growing number of sets in the sliding window, i.e., with the window duration and the stream rate.

In summary, we make the following key contributions:
\begin{itemize}[noitemsep,topsep=0pt,parsep=0pt,partopsep=0pt]
\item We present SWOOP, a novel algorithm for continuous \mbox{top-$k$} set similarity joins over streams. Two salient features of SWOOP are (1) the efficient generation of candidates when new sets enter the sliding window and (2) the incremental maintenance of a minimal stock to deal with expiring sets.
\item We introduce the concept of a well-behaved similarity function to accurately characterize the applicability of SWOOP.
\item We present a solution to contend with the absence of so-called token frequency maps in streams; we particularly target difficult streams with very skewed token distributions.
\item We report on empirical studies showing that SWOOP is capable of running orders 
of magnitude faster than the state of the art.
\end{itemize}

\paragraph{Outline} Section~\ref{sec:problemdef} formulates the problem. Section~\ref{sec:indbaseline} introduces the stream join framework and a baseline solution. Section~\ref{sec:supsimfunctions}  defines well-behaved similarity functions. Section~\ref{sec:invlist} explains the candidate generation algorithm, including the handling of difficult datasets. Section~\ref{sec:skyband} covers the maintenance of the join result. Section~\ref{sec:experiments} reports on the empirical study. Section~\ref{sec:relwork} covers related work, and Section~\ref{sec:conclusion} concludes the paper.

\section{Problem Setting and Definition}
\label{sec:problemdef}

\paragraph{Basic Concepts}
A \emph{stream} $R$ is a sequence of two-tuples $(r_i, t_i)$, where $r_i$ is a set and $t_i$ is a timestamp. The $i$-th tuple in $R$ is denoted as $R_i$. The timestamp is monotonically increasing with the sequence number, i.e., for any two tuples $R_i=(r_i,t_i)$ and $R_j=(r_j,t_j)$, $i<j\Rightarrow t_i\leq t_j$.
A \emph{sliding window} $W$ of duration $w$ over stream $R$ contains all tuples of $R$ that are no older than $w$: $W=\{(r_i,t_i)\in R \mid \indextime - w < t_i \leq \indextime \}$, where $\indextime$ is the current time, also refered to as the \emph{index time}. The sets in the sliding window are called $\emph{valid}$.
Table~\ref{tab:notation} summarizes the notation.

\begin{table}[htbp]
 \centering
 \begin{tabular}{@{}c@{}c@{}}
 \begin{tabular}{@{}r|l@{}}
   $R$         & stream of timestamped sets \\
   $r_i$       & $i$-th set in stream  $R$ \\
   $t_i$       & timestamp of set $r_i$ \\
   $W$         & sliding window on $R$ \\
  $w$          & window duration (time) \\
  $\indextime$ & index time (also: current time) \\
\end{tabular}
\hspace{-2em}\,
  \begin{tabular}{@{}r|l@{}}
  $T$            & top-$k$ list \\
  $p$ & pair of sets \\
  $sim(p)$ & similarity of sets in $p$ \\
  $e_p$          & end time of pair $p$ \\

  $\uptau$       & set similarity threshold \\
  $l_r$          & cardinality of set $r$ \\
 \end{tabular}
 \end{tabular}
 \caption{Notation.}%
 \label{tab:notation}%
\end{table}

\paragraph{Window Join} To simplify the presentation, we discuss a self join scenario, where a stream is joined with itself; with minor modifications, all techniques presented in this paper extend to the general case of joining two different streams.

%
%

The top-$k$ set similarity join in sliding window $W$ returns the $k$ most similar pairs of sets from stream $R$ that are valid at the time the query is issued. Various functions have been proposed to assess the similarity between sets, e.g., Jaccard, Cosine, or Dice~\cite{conf/icde/XiaoWLS09}.

\begin{defn}[One-Time Top-$k$ Set Similarity Join]

 Given a sliding window $W$ over stream $R$ and a set similarity function $set\_sim(\cdot,\cdot)$, the \emph{one-time top-$k$ set similarity join} returns a list of $k$ set pairs $T = \langle p_1, p_2, \ldots , p_k \rangle$ from $R \times R$, such that (1) each pair $p_x$ is composed of valid sets, (2) $T$ is ordered descendingly according to $set\_sim(\cdot,\cdot)$, (3) for all $(r_i,r_j)\in T$, $i>j$, (4) for all $(r_i,r_j)\in T$, $set\_sim(r_i,r_j)>0$, (5) for all pairs $(s_i,s_j)$ of valid sets in $R \times R$ not in $T$,  $set\_sim(s_i,s_j)\leq  min_{(r_i,r_j)\in T} set\_sim(r_i,r_j)$. Finally, $T$ may contain fewer than $k$ pairs if fewer than $k$ pairs qualify.

\end{defn}
In the definition, condition~3 eliminates symmetric pairs such that only one of $(r_i,r_j)$ and $(r_j,r_i)$ is included in $T$.

The above join is a one-time query because it is executed once. We consider the continuous variant of the query that maintains an up-to-date result from when it is started until when it is stopped.
As time passes, sets leave window $W$ (expire), and new sets enter $W$. The join result $T$ must be kept up-to-date when such events occur. 
A set $r_i$ that enters window $W$ at time $t_i$ forms a new pair with all other sets $r_j$ in $W$, where $j<i$. A new pair enters the join result if it is sufficiently similar.  
When a set $r_i$ leaves $W$ and thus expires, all pairs that contain $r_i$ become invalid. Invalid pairs must be removed from $T$, and they must be replaced by valid pairs. In general, a pair $(r_i,r_j)$ is valid
from time $\max(t_i,t_j)$ (when the younger set enters the window) until time
$\min(t_i,t_j)+w$ (when the older set leaves the window). Since we only consider pairs $(r_i,r_j)$ with $i>j$, the validity interval is always $[t_i, t_j+w)$.

Valid pairs always have their start time in the sliding window (time period $\{t\mid\indextime-w< t \le \indextime\}$) and their end time in the so-called \emph{future window} (time period $\{t\mid\indextime < t \le \indextime+w\}$), i.e., their validity interval contains $\indextime$. Invalid pairs have both their start and end time in the sliding window. This is illustrated in Figure~\ref{fig:2windows}.

\begin{figure}[htbp]
 \centering
  \myasygraphics{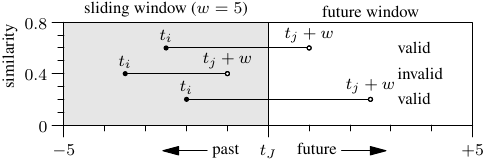}
 \caption{Valid and invalid pairs, sliding and future window.}
 \label{fig:2windows}
\end{figure}
\paragraph{Problem Statement}
Our goal is to solve the \emph{continuous top-$k$ set similarity join}  over rapid streams using a sliding window.

\section{Join Framework and Baseline}
\label{sec:indbaseline}

\paragraph{Stream join framework}  We introduce our stream join framework, illustrated in Figure~\ref{fig:algmodel}, and cover a baseline implementation of the framework. The framework comprises three constructs:

\begin{itemize}[noitemsep,topsep=0pt,parsep=0pt,partopsep=0pt]
\item \emph{Index time $\indextime$}  is the current time in the framework and defines the sliding window. All data structures in the framework must be up-to-date w.r.t.\ the index time.

 \item \emph{Stock $\stock$} maintains the join result $T$ at time $\indextime$ and additional, valid pairs to deal with expiring sets.

 \item \emph{Window $W$} stores all tuples in stream $R$ covered by the sliding window at time $\indextime$. $W$ is used when evaluating the similarity between pairs of sets and when expiring sets as the index time increases (i.e., the sliding window is advanced).

\end{itemize}

The framework supports three operations: (i) $topk()$ retrieves the join result $T$ at index time $\indextime$; (ii) $set\_index\_time(t)$ sets the index time to $t\geq \indextime$; (iii) $insert(r_i,t_i)$ sets the index time to $t_i\geq \indextime$ and inserts a new set $r_i$ into the index. Sets must be inserted in the order of their appearance in $R$. The index time can never decrease.

\paragraph{Baseline} The baseline algorithm implements stock $\stock$ as a binary tree ordered by descending similarity of the pairs, i.e., the top-$k$ pairs are ranked first. Window $W$ is implemented as a FIFO queue that can be iterated and supports the usual peek/pop/push operations.
We discuss the three operations in the join framework.

(i) \emph{topk()} retrieves the join result $T$ at index time $\indextime$ by traversing the first $k$ pairs in stock $\stock$ (or $|\stock|$ pairs if $|\stock|<k$). No index update is required.

(ii) \emph{$set\_index\_time(t)$} updates the index time $\indextime$ and pops all sets from window $W$ that expire when the sliding window is advanced ($(r_i,t_i)\in W$ where $t_i\leq \indextime-w$). The corresponding entries $(r_x,r_i, \uptau, e_p)\in\stock$ with $e_p\le \indextime - w$ are deleted. 
%

(iii) \emph{$insert(r_i,t_i)$}  first advances the sliding window to position $t_i$ and updates the affected data structures such that $W$ only contains valid pairs ($set\_index\_time(t_i)$). Next, the similarity of each pair $(r_i,r_j)\in\{r_i\}\times \{r_j\mid (r_j, t_j) \in W\}$ is computed; if $set\_sim(r_i,r_j)>0$, the pair is a candidate and is ranked in stock $\stock$. After the insert, $\stock$ contains the join result as of time $t_i$.

\begin{figure}[t]
 \centering
  \myasygraphics{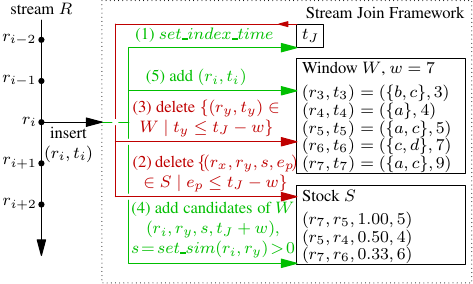}
 \caption{Inserting a new set into the stream join framework.}
 \label{fig:algmodel}
\end{figure}

Figure~\ref{fig:algmodel} illustrates $insert(r_i,t_i)$ for an incoming  two-tuple $(r_i,t_i)$ from stream $R$. Steps~1--3 reflect the call to $set\_index\_time(t_i)$, which (1) updates $\indextime$, (2) removes invalid pairs from stock $S$, and (3) removes expired sets from window $W$.  Step (4) adds the new pairs generated by $get\_candidates(r_i)$ to $S$. Step (5) adds set $r_i$ to $W$. 

\emph{Complexity of baseline.}\hspace{0.8mm}{\newdimen\origiwspc\origiwspc=\fontdimen2\font\fontdimen2\font=0.5ex Stock $\stock$ is of size $O(|W|^2)$
and dominates the memory complexity. The insert operation runs in $O(|W| \log|W|)$ time since a new set must be paired with every other set in $W$, and the pairs must be inserted into binary tree $\stock$. Function $set\_index\_time$ scans the stock in time $O(|W|^2)$ for expiring sets; removing a set has cost $O(\log S)=O(\log W)$. Finally, $topk$ runs in optimal $O(k)$ time.\fontdimen2\font=\origiwspc}

\paragraph{Solution overview} The inefficiency of the baseline solution arises from the many candidate pairs generated for each incoming set and the quadratic size of the stock, which must be maintained under frequent changes. We address these issues in the following sections. 
The next section characterizes the scope of our solution.
Section~\ref{sec:invlist} introduces an efficient technique to generate candidates: using an index on tokens together with an upper and a lower bound, only a small fraction of the sets in window $W$ needs to be considered. 
Section~\ref{sec:skyband} proposes an efficient stock implementation that  stores only $O(k\cdot |W|)$ pairs, is maintained incrementally, and supports efficient lower bound queries.

\section{Supported Similarity Functions}
\label{sec:supsimfunctions}

Our solution works with the most common similarity functions, including Jaccard, Cosine, Dice, Overlap, and Hamming distance, but is not limited to these functions. We introduce the concept of a \emph{well-behaved set similarity function} to abstract the applicability from similarity functions and instead identify the essential properties that a similarity function must satisfy to work with our solution.

\begin{defn}[Well-behaved similarity function]
\label{def:sim-dev}
  A similarity function between two sets, $set\_sim(r,s)$, is \emph{well-behaved} iff there is a function $sim(l_r,l_s,o)=set\_sim(r,s)$ that only depends on the set lengths $l_r=|r|$, $l_s=|s|$, and the overlap $o=|r\cap s|$, and the following properties hold:
  \begin{enumerate}[noitemsep]
    \item $sim(l_r,l_s,0)=0$ \label{item:sim-def-first}
    \item $sim(l_r,l_s,o)=sim(l_s,l_r,o)$ (symmetry)
    \item $sim(l_r,l_s,o)$  monotonically increases with increasing overlap $o$ ($l_r,l_s$ are fixed) \label{item:sim-def-mono-overlap}
    \item $sim(l_r,l_s,o)$ monotonically increases with increasing overlap $o=l_s$, i.e. $s\subseteq r$ ($l_r$ is fixed) \label{item:sim-def-mono-subset}\label{item:sim-def-last}
    \item there is a function $\eqoverlap(l_r,l_s, \uptau)$ that computes the minimum required overlap $o$ such that $sim(l_r,l_s,o)\geq \uptau$ \label{item:sim-def-eqoverlap}
  \end{enumerate}
\end{defn}

 \begin{lem}
   Jaccard, Cosine, Dice, and Overlap similarity, and the Hamming distance are well-behaved set similarity functions.
 \end{lem}

 \begin{proof}
Table~\ref{tab:simdef} defines functions $sim(l_r,l_s,o)$ for the similarity and distance functions. Claims~\ref{item:sim-def-first}--\ref{item:sim-def-last} are easily verified using these definitions. Next, the table provides definitions of $\eqoverlap(l_r,l_s, \uptau)$, which is computed by solving the inequality $sim(l_r,l_s,o)\geq \uptau$ for $o$, from which claim~\ref{item:sim-def-eqoverlap} follows.
 \end{proof}

\begin{table}[htb]
 {\small
  \begin{tabular}{l|c|c|c}
   Similarity & $set\_sim(r,s)$                               & $sim(l_r,l_s,o)$                & $\eqoverlap(l_r,l_s,\uptau)$         \\\hline\hline
   Jaccard    & $\frac{|r\cap s|}{|r\cup s|}$            & $\frac{o}{l_r+ l_s-o}$          & $\frac{\uptau}{1+\uptau} (l_r+ l_s)$ \\\hline
   Cosine     & $\frac{|r\cap s|}{\sqrt{|r|\cdot  |s|}}$ & $\frac{o}{\sqrt{l_r\cdot l_s}}$ & $\uptau\sqrt{l_r\cdot  l_s}$         \\\hline
   Dice       & $\frac{2\cdot(|r\cap s|)}{|r|+ |s|}$     & $\frac{2\cdot o}{l_r+ l_s}$     & $\frac{\uptau(l_r+ l_s)}{2}$         \\\hline
   Overlap    & $|r\cap s|$                              & $o$                             & $\uptau$
   \\\hline
   Hamming    & $|(r\cup s) \setminus (r\cap s)|$        & $l_r+ l_s-2\cdot o$             & $\frac{l_r+ l_s-\uptau+1}{2}$
   \\
  \end{tabular}
 }
 \caption{Examples of well-behaved similarity functions.}
 \label{tab:simdef}
\end{table}


\section{The Candidate Index}
\label{sec:invlist}


\subsection{Overview}
\label{sec:invlist:overview}

We discuss the efficient generation of candidates in SWOOP. Candidates are pairs that must be inserted into the stock. We use an inverted list index, the \emph{candidate index} $I$, to compute candidates. The keys in the index are tokens, and the values are lists of all valid sets in which the token appears. When a new set $r_i$ enters the sliding window, the lists of all tokens in $r_i$ are accessed to retrieve candidates, and index $I$ is updated. Efficient index updates are discussed in Section~\ref{sec:update_inverted_index}.

A naive use of an inverted list index offers little improvement over the baseline: only the set pairs with no overlap are avoided, and the use of the index tends to cause more cache misses than the baseline. In static scenarios, all sets are known up front and are preprocessed to support efficient indexing and effective candidate filters. For example, the tokens within a set are sorted by increasing frequency (to favor the prefix filter~\cite{SCVGRK06}), the sets are processed and indexed in non-decreasing length order (to support the length filter~\cite{RBYMRS07,maau14}), and sets need not be removed as the index size is bound by the data size. In our streaming scenario, we cannot preprocess the data, and our index must support efficient updates as new sets arrive and old sets expire. We propose candidate filters applicable to streams that effectively prune candidate sets which cannot contribute to the join result. 

\paragraph{Filters} The \emph{positional upper bound} filter introduced in Section~\ref{sec:positional-upper-bound} is based on the lookup position $\rho$ of a token in $r_i$ with the following reasoning: if a potential candidate $r_j$ is first encountered in the $\rho$-th list, there must be at least $\rho-1$ tokens in $r_i$ that do not exist in $r_j$.
The \emph{skyband lower bound} filter discussed in Section~\ref{sec:stock-lower-bound} is derived from the pairs that are already in the stock. A potential candidate pair is called \emph{irrelevant} and can be discarded if its not sufficiently similarity to be part of the top-$k$ result at any time in the future. We derive this minimum required similarity by inspecting the stock and taking into account the end time of the candidate pair under consideration.

\paragraph{Candidate Generation} In Section~\ref{sec:efficient-candidate-generation}, we devise a new candidate generation algorithm that uses our filters and the candidate index. 
Figure~\ref{fig:insertmodel} illustrates the algorithm for a newly inserted example set $r_7=\{a,c\}$ with timestamp $t_7=9$. The candidates are computed as follows. (1) A lookup of the tokens of $r_7$ in the candidate index $I$ returns two lists. (2) The lists are scanned from tail to head and produce so-called \emph{pre-candidates} (shaded in gray) until our filters tell us to stop (cropping). (3) We compute the similarity of each (deduplicated) pre-candidate pair and apply the skyband lower bound to prune irrelevant pairs.
The resulting candidates are collected in $C$. A candidate is a pair with its similarity and its end time. (4) Index $I$ is updated with the new tokens of set $r_i$ (dashed frame). (5) The stock is updated with the candidates in $C$ (dashed frame).

Section~\ref{sec:opttoken} deals with token orders and discusses the lookup order of tokens in the candidate index. 

\begin{figure}[htb]
 \centering
  \myasygraphics{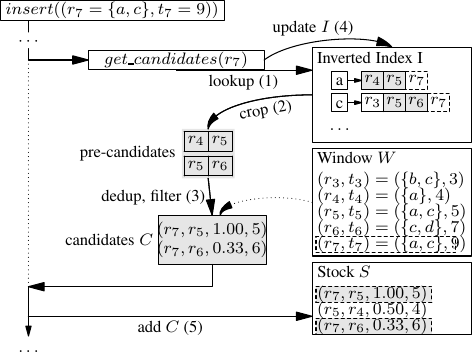}
 \caption{Efficient candidate generation.}
 \label{fig:insertmodel}
\end{figure}

\subsection{Updating the Candidate Index}
\label{sec:update_inverted_index}

Since only valid sets are indexed, index $I$ must be updated frequently. In particular, we must update $I$ when old sets expire and when new sets enter sliding window $W$.

We implement the candidate index with doubly-linked lists, the \emph{index lists}, and keep the sets in the lists ordered increasingly by their expiration time. This allows us to efficiently remove expiring sets from the heads of the lists. The list order comes for free: The timestamps of the new sets cannot decrease; thus, we append new sets to the tails of the relevant lists. A set $r$ is inserted/deleted in $O(|r|)$ time, independently of the list length. Figure~\ref{fig:invlistindex} illustrates the index update for an expiring set $r_2$ and a new set $r_6$.

\begin{figure}[htb]
 \centering
 \myasygraphics{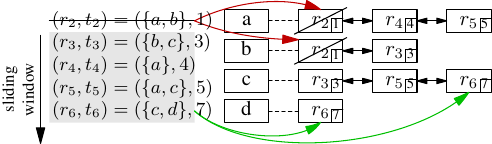}
 \caption{Candidate index: insertion and deletion.}
 \label{fig:invlistindex}
\end{figure}

As a convenient side effect of the list order, we retrieve the candidate pairs  in sort order of their expiration time: A lookup of $r_i$ returns all lists $I(v)$ with tokens $v\in r_i$. Let some $r_j\in I(v)$ form a candidate pair $(r_i,r_j)$ with $r_i$. The expiration time of the candidate pair is $t_j+w$, i.e., it depends only on set $r_j$. Thus, the list order propagates to the candidate pairs.

\subsection{Positional Upper Bound}
\label{sec:positional-upper-bound}

We derive an upper bound on the set similarity that will be used to prune candidates during lookups in index $I$.

 \begin{theo}
   \label{th:positional-upper-bound}
   Given a well-behaved similarity function $set\_sim(\cdot,\cdot)$, sets $r$ and $s$. If at least $i$ tokens in $r$ do not exist in $s$, then the following upper bound on the similarity between $r$ and $s$ holds:
   \[
   set\_sim(r,s) \leq sim(|r|,|r|-i,|r|-i)
   \]
 \end{theo}

\begin{proof}
    We need to show that $sim(l_r,l_s,o)$ is maximum if $|s|=|r|-i$ and overlap $o=|r\cap s|=|r|-i$. W.l.o.g.\ assume $|s|\leq |r|$. For the case $s\subseteq r$, the similarity is maximized for the maximum size of $o=|s|$ (Def.~\ref{def:sim-dev}, claim (\ref{item:sim-def-mono-subset})). For given set lengths $|r|$ and $|s|$, the similarity is maximum if $s\subseteq r$ since $o<|s|$ in all other cases (Def.~\ref{def:sim-dev}, claim (\ref{item:sim-def-mono-overlap})). Thus, the maximum similarity is achieved when $|s|=o=|r|-i$.
\end{proof}

Consider a lookup of set $r$ in the index $I$. The lookup returns a list $I(v)$ for each token $v\in r$. Let $v_\pos$ be the $\pos$-th token of set $r$ that we look up in $I$; we call $\pos$ the lookup position. A set $s\in I(v_\pos)$ is \emph{new} if $\pos=1$ or $s\notin I(v_q)$ for $1\leq q < \pos$. For the new sets $s\in I(v_\pos)$, we know that there are at least $\pos-1$ tokens in $r$ that do not exist in $s$. Based on Theorem~\ref{th:positional-upper-bound} we derive the  following \emph{positional upper bound}:
\[ ub(|r|,\pos)=sim(|r|,|r|-\pos+1,|r|-\pos+1). \]

For any new set $s\in I(v_\pos)$, $set\_sim(r,s)\leq ub(|r|,\pos)$. This principle has been used before in the context of a specific set similarity function (e.g., Jaccard)~\cite{conf/icde/XiaoWLS09}. Compared to previous work, we provide a formal proof, do not require a global order of tokens, and generalize the bound to the class of well-behaved set similarity functions.

Figure~\ref{fig:maxsim} illustrates the upper bound for the Jaccard similarity on a set of length $|r|=5$.

\begin{figure}[htb]
 \centering
 \myasygraphics{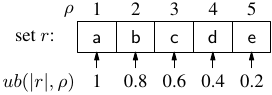}
 \caption{Positional upper bound for Jaccard.}
 \label{fig:maxsim}
\end{figure}

\subsection{Skyband Lower Bound}

\label{sec:stock-lower-bound}

We define the skyband lower bound that, together with the positional upper bound from the previous section, allows us to stop processing an index list early.  The skyband lower bound marks the boundary of the so-called \emph{skyband}, which is formed by the $k$ most similar pairs at any time $t> \indextime$ in the future; thereby, only pairs that exist at index time $\indextime$ are considered. The skyband is maintained in stock $S$. The red staircase functions in Figure~\ref{fig:obsolete_interval_simple} show the skyband lower bound for two example stocks.

  \begin{figure}[htb]
   \centering
    \subfigure[$p_1$ is relevant]{
    \myasygraphics{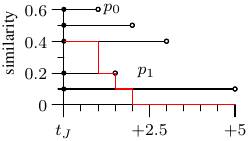}%
    \label{fig:obsolete_interval_simple_a}
    }%
    \subfigure[$p_1$ is irrelevant]{
    \myasygraphics{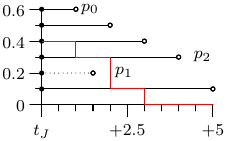}
    \label{fig:obsolete_interval_simple_b}
    }
   \caption{Skyband lower bound (red line) ($k=3$).}
   \label{fig:obsolete_interval_simple}
  \end{figure}

The \emph{skyband lower bound}, $lb(t,k)$, is defined as the similarity of the $k$-th pair at time $t> \indextime$ in stock $S$. The efficient computation of $lb(t,k)$ is discussed in Section~\ref{sec:efficient-lower-bound-computation}. We next introduce the concept of \emph{irrelevant pairs}, which need not be considered as candidates. Then we show how to detect irrelevant pairs using the lower bound.

\paragraph{Irrelevant Pairs} A pair $p = (r_i,r_j)$ is irrelevant if it is not part of the join result $T$ at index time $\indextime$ and will never become part of $T$. This is the case if for the remaining life time of the pair, $[\indextime,t_j+w)$, at least $k$ more similar pairs exist.

Irrelevant pairs are identified by considering their rank at their end time. The pair $p$ is irrelevant if the rank of $p$ at its  end time exceeds $k$, i.e., at least $k$ pairs exist that are better than $p$ for the whole remaining life time of $p$. Note that pairs inserted in the future can never increase the rank of $p$.

A pair may (a) be irrelevant before it is inserted into stock $S$ (then we can avoid inserting it), or (b) it may become irrelevant due to the insertion of another pair.

\begin{examp}
 Consider pair $p_1$ in Figure~\ref{fig:obsolete_interval_simple_a} with $sim(p_1)=0.2$ and end time $e_{p1}=\indextime+1.5$. For $k=3$, $p_1$ is relevant since the rank at its end time is $3\leq k$. The rank at index time $\indextime$ is $4$; the rank improves to $3$ at time $\indextime+1$ when $p_0$ becomes invalid. If we insert pair $p_2$, $p_1$ becomes irrelevant as illustrated in Figure~\ref{fig:obsolete_interval_simple_b}: the rank at its end time is now $4>k$.
 New pairs cannot improve the rank of pairs that are already in the stock; at best, they leave it unchanged. 

\end{examp}

\paragraph{Detecting Irrelevant Pairs} 
We use the skyband lower bound to identify irrelevant pairs. A pair $(r_i,r_j)$ with end time $t=t_j+w$ is irrelevant iff its similarity is below the lower bound at its end time $t$:
$(r_i,r_j)\text{ is irrelevant } \Leftrightarrow set\_sim(r_i,r_j)< lb(t,k)$.

\begin{lem}
 \label{lem:simkth} The skyband lower bound, $lb(t,k)$, is a non-increasing function in $t$.
\end{lem}

\begin{proof}
 All pairs start at or before the index time. The $k$-th pair $p=(r_i,r_j)\in T$ at index time has similarity $\tau=set\_sim(r_i,r_j)$ and end time $t=t_j+w$. When a pair $p\in T$ ends, a pair $p_i$ with similarity at most $\tau$ is promoted to position $k$ in $T$. Thus, the skyband lower bound cannot increase.
\end{proof}

\subsection{Efficient Candidate Generation}
\label{sec:efficient-candidate-generation}

We use the positional upper bound and the skyband lower bound to efficiently prune candidates during the lookup in index $I$, as illustrated in Figure~\ref{fig:listprocessing}. Recall that the positional upper bound, $ub(|r_i|,\pos)$, is constant for an index list $I(v_\pos)$, where $v_\pos\in r_i$ is the $\pos$-th token that we look up in the index (blue line in the figure). The skyband lower bound, $lb(t,k)$, on the other hand, depends on the time $t$ (red line segments).

\begin{figure}[htb]
 \centering
  \myasygraphics{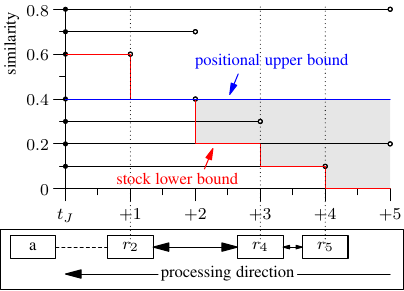}
 \caption{List processing with bounds, $k=3$.}
 \label{fig:listprocessing}
\end{figure}

The similarity of any pair $(r_i,r_j)$ formed with an entry $r_j$ in the index list $I(v_\pos)$ falls on or below the blue line. A pair is relevant iff its end point is on or above the red line. Thus, a pair with a set from list $I(v_\pos)$ is relevant iff its end point falls into the gray region in Figure~\ref{fig:listprocessing}.

More specifically, we employ the bounds as follows. We process the index list $I(v_\pos)$ from tail to head such that the end times $t=t_j+w$ of pairs $(r_i,r_j)$ formed with the sets $r_j\in I(v_\pos)$ do not increase (cf.\ Section~\ref{sec:update_inverted_index}). For each pair, we compute the lower bound at its end time $t$. We stop processing the list when having formed a pair with a lower bound above the upper bound, i.e., $lb(t,k)>ub(|r_i|,\pos)$. This is correct due to Lemma~\ref{lem:simkth}: the lower bounds of all remaining pairs will also exceed the upper bound threshold, i.e., no additional relevant pairs can be formed.

Algorithm~\ref{alg:newsets} generates candidate pairs for a new set $r_i$ using candidate index $I$. The basic structure is as follows (cf.\ Figure~\ref{fig:insertmodel}): for each token of the new set, $r_i[\rho]$, we probe $I$ to get a list of set IDs. The list is cropped, i.e., traversed from tail to head in line~\ref{algline:reverse} until the stopping condition based on our upper and lower bounds holds. The list elements are called pre-candidates and are stored with their lower bound in hashmap $M$. In the next step (lines \ref{algline:candbegin}--\ref{algline:candend}), we verify the pairs by computing their overlap to get the final set of candidates.  Finally, the new set $r_i$ is inserted into the index.

\begin{algorithm}[htb]
 \SetKwInOut{Inp}{Input}
 \SetKw{Break}{break}
 \SetKwProg{Fn}{Function}{}{end}
 \Inp{$r_i$: set to be looked up in candidate index $I$}
 \Fn{get\_candidates($r_i$)}{
 \tcp{get pre-candidate pairs from index}
  $M:$ empty candidate map, key:\,$r_j$,\,val:\,lower\_bound\nllabel{algline:getcandinit}\;
    \For{ $\pos$ in $1$ to $|r|$\nllabel{algline:newsetforloop}}{
    $upper\_bound\gets ub(|r_i|,\pos)$\;
   \tcp{traverse one index list}
   \ForAll{$r_j$ in reverse order of $I(r_i[\pos])$\nllabel{algline:reverse}}{
    $lower\_bound\gets lb(t_j+w,k)$\nllabel{algline:begininvok}\;
    \lIf{$lower\_bound>upper\_bound$\nllabel{algline:listupper}}{
     \Break
    }
    $M[r_j] \leftarrow lower\_bound$\nllabel{algline:end1invok}\;
   }
  }
  \tcp{compute candidates for insertion in stock}
  $C\gets \emptyset$\;
  \For{$(r_j,lower\_bound)$ in $M$\nllabel{algline:candbegin}}{
   $\uptau_o \leftarrow \eqoverlap(|r_i|,|r_j|,lower\_bound)$\;
   \If{$|r_i\cap r_j| \ge \uptau_o$}{
    $C\leftarrow C\cup \{(r_i,r_j,sim(|r_i|,|r_j|,|r_i\cap r_j|),t_j\!+\!w)\}$\nllabel{algline:candend}\;
   }
  }
  \tcp{update candidate index}
  \lFor{ $\pos$ in $1$ to $|r|$\nllabel{algline:newsetforloop2}}{
   $I(r_i[\pos])\leftarrow I(r_i[\pos])\circ(r_i)$
  }
  \Return{C}\;
 }
 \caption{Get candidates from index $I$.}
 \label{alg:newsets}
\end{algorithm}

A candidate pair $(r_i,r_j)$ is verified by checking $|r_i\cap r_j|\geq \uptau_o$. The overlap computation stops early when $\uptau_o$ cannot be reached. As shown by Mann et al.~\cite{mann2016} for threshold-based set similarity joins, stopping early has a major impact on the performance.

A pre-candidate $r_j$ may appear in multiple lists. Since $lower\_bound$ for $r_j$ does not change during a $get\_candidates()$ call, we look up the bound in $M$ and need not recompute it (line~\ref{algline:begininvok}).

\subsection{Optimized Token Processing Order}
\label{sec:opttoken}

Before we process a new set $r_i$, we order its tokens. This is required for the merge-like overlap computation. A well-known approach is to order sets by decreasing token frequency, i.e., rare tokens appear earlier in the sorted sets. This is useful in two ways: First, rare tokens have short lists in the index, which we leverage as discussed below. Second, the stop condition in the merge-like overlap computation improves with the number of mismatches, which are more likely for rare tokens.

Processing rare tokens (i.e., short lists) first when we retrieve candidates for $r_i$ has a substantial impact on the performance. This is due to our upper bound, which improves with the lookup position of a token. A tighter upper bound allows us to skip a longer section of the index list. Thus, we want to process long lists as late as possible and use the bound to skip large fractions of the long lists.

Non-streaming set similarity joins count the frequency of each token in a preprocessing step and establish the order up front. This is not possible in our setting since the sets arrive on a stream and are not known up front. Instead, we number each token when it first appears in the stream. Then, a new set is sorted in descending order of the first occurrence of its tokens, i.e., tokens that occur later are sorted lower in sort order. The idea is that frequent tokens are more likely to occur earlier in the stream than infrequent ones. 

In our experiments, we show that our ordering heuristic is effective if the token distribution is stable over time, i.e., a token appears with the same probability in each subsection of the stream. Unfortunately, some real world data does not satisfy this assumption. This leads to inefficiencies if we process the tokens in the order of their sort position (as in Algorithm~\ref{alg:newsets}, line~\ref{algline:newsetforloop}).
To deal with skewed token distributions, we process a new set $r_i$ as follows: We first retrieve the index lists of all tokens of $r_i$ and heapify the lists such that the shortest list is on top of the heap. We then pop the lists and process them until the heap is empty. This approach substitutes the order in Algorithm~\ref{alg:newsets}.


\section{Maintaining the Join Result}
\label{sec:skyband}

The stock $S$ maintains the join result. This includes ranking the $k$ most similar pairs at index time $t_J$ and keeping enough valid replacements for result pairs that leave the sliding window and thus become invalid. We require the following functionality.

\begin{itemize}[noitemsep,topsep=0pt,parsep=0pt,partopsep=0pt]

  \item $topk(k)$: Return the top-$k$ result at index time $\indextime$.
  
  \item $set\_index\_time(t), t\geq\indextime$: Increase the index time to $t$ and remove expiring pairs.
  
  \item $lb(t,k)$: Get the skyband lower bound at time $t$, i.e., the similarity of the $k$-th pair at time $t> \indextime$.
  
  \item $insert(C)$: Insert a collection of candidate pairs $C$ that all start at index time $\indextime$.
  
\end{itemize}

The $topk$ operation is trivial: it traverses the first $k$ elements of $S$ in sort order. The other operations are discussed below.

\paragraph{Stock Data Structure} 
For a pair $p=(r_i,r_j)$, the stocks stores a quadruple $(r_i,r_j,sim(p), e_p)$, where $sim(p)$ is the similarity of the pair and $e_p$ is its end time.
We implement $S$ as a binary search tree ordered by decreasing similarity (and lexicographically by descending end time, ascending $i$ and $j$ to break ties). 

In addition to search, two rank operations are supported in $O(\log |S|)$ time (cf.\ Section~\ref{sec:experiments}): (1) given an item $p\in S$, the rank of $p$ in the sort order is computed; (2) given rank $i$, the $i$-th item $p\in S$ in the sort order is returned. In our algorithms, we use the notation $S[i]$ to access the $i$-th item of $S$ in sort order.

\subsection{Incrementing the Index Time}
\label{sec:incrementing-the-index-time}

The $set\_index\_time$ operation advances the sliding window and removes expiring pairs from stock $\stock$. If pairs from the current join result $T \subseteq \stock$ are removed, they must be replaced by other pairs. The baseline algorithm keeps all valid pairs as potential replacements. As we will show, this is not necessary.

\paragraph{Minimal Stock}  We call stock $\stock$ correct if it contains all
pairs that may be required in the future to maintain $T$, i.e., all pairs that are relevant at index time $\indextime$ (cf.\ Section~\ref{sec:stock-lower-bound}). 
We call $\stock$ \emph{minimal} if it is correct and removing any pair makes it incorrect. 
The stock maintained by the baseline, which is correct but not minimal, is quadratic in the window size $|W|$. The minimal stock is linear in $|W|$.

\begin{lem}
  The size of a minimal stock $\stock$ is $O(k\cdot |W|)$.
\end{lem}

\begin{proof}
 The deletion of a set $r_j$ invalidates at most $k$ pairs $(r_i,r_j)$ in $T$ since $|T|\le k$ ($|T|< k$ if fewer than $k$ pairs have non-zero similarity). The worst case is illustrated in Figure~\ref{fig:worstcase}, where $k=3$ pairs $(r_i,r_1)$ end at time $\indextime+1$ and must be replaced by the next $k$ pairs in the similarity order. Since only $|W|$ valid sets can expire, no more than $k\cdot |W|$ replacements are required.
 \end{proof}

 \begin{figure}[htb]
  \centering
  \myasygraphics{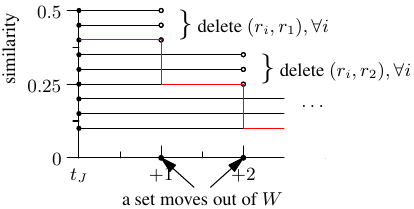}
  \caption{Worst case, $k=3$.}
  \label{fig:worstcase}
 \end{figure}

\paragraph{End Time Index}
Function $set\_index\_time(t)$ removes all pairs $p\in\stock$ with end time $e_p$ smaller than $t$. The naive solution scans $\stock$, checks the end time of each pair, and removes expired pairs. For $n\leq |S|$ expired pairs, the runtime is $O(|\stock|+n\log |\stock|)$. This is too slow as the index time is potentially incremented by each new set in the stream.

We introduce the \emph{end time index} $E$ that maintains the same elements as stock $\stock$, but orders them by ascending end time (ascending similarity, descending $i$, $j$ for pair $p=(r_i,r_j)$). Like $\stock$, $E$ is implemented as a binary tree that supports rank operations in logarithmic time. Index $E$ is updated whenever $\stock$ is updated, thus $|E|=|S|$.

Our implementation of $set\_index\_time(t)$ scans the end time index only while the end time $e_p$ is below $t$. Then the scan stops, and the remaining pairs are not touched. Each scanned pair is removed. The removal of $n\leq |S|$ invalid pairs takes $O(n\log |\stock|)$ time. Since each pair can be removed only once, the worst case $n=|S|$ is infrequent, and the average complexity is $O(\log |\stock|)$.

\subsection{Efficient Lower Bound Computation}
\label{sec:efficient-lower-bound-computation}

The skyband lower bound $lb(t,k)$ (cf.\ Section~\ref{sec:stock-lower-bound}) is the similarity of the $k$-th pair in $S$ at some future time $t>\indextime$. It is used during candidate generation and is evaluated for each entry in the index lists until the stopping condition is reached.

A straightforward implementation scans $S$ and returns the $k$-th pair $p$ at
time $t$ that satisfies $e_p \ge t$. This takes $O(|S|)$ time, which is too
expensive since the lower bound needs to be computed for each pre-candidate.
We exploit the fact that $\stock$ is minimal and use the end time index $E$ to retrieve the $k$-th pair at time $t$. The following theorem establishes a connection between $E$ and $S$ that is leveraged for the efficient computation of the skyband lower bound.

%

\begin{theo}
 \label{lem:skybandpoints}
Let $t\geq t_J$ be a timestamp, $p$ the pair in $E$ with the smallest timestamp such that $e_p\geq t$, and $v$ the rank of $p$ in endtime index $E$. If stock $S$ is minimal, then the $k$-th pair in $S$ at time $t$ is $S[k+v-1]$.
\end{theo}

\begin{proof}
 By induction on $v$. Pair $p=E[1]$ covers the interval $\indextime \leq t < e_p$ and is the first pair to end; in this interval, the $k$-th pair in $S$ is $S[k+v-1]=S[k]$. 
 Assumption: The $k$-th pair in $S$ during the interval $[t, e_p)$ is $S[k+v-1]$. 
 Note that $p$ is in the top-$k$; otherwise $p$ could be removed (which is not possible in a minimal stock). Assume unique end times in $E$: The pair $E[v+1]$ defines the next interval. Since $p$ is now invalid, the next element in the stock, $S[k+v]$, is promoted to become the $k$-th pair in $S$. 
 Now assume the general case of $n$ entries in $E$ with the same end time: $v$ is always the position of the first of these entries in $E$. The pair $E[v+n]$ defines the next interval, invalidating the former top-$k$ entries $E[v]$ to $E[v+n-1]$ and promoting $S[k+v-1+n]$ to rank $k$ in $S$. 
  \end{proof}


To compute $lb(t,k)$, we search $E$ for the smallest pair (in sort order)  with $e_p\geq t$ and retrieve its rank $v$. Operation $lb(t,k)$ is the similarity of the pair at position $v+k-1$ in $S$. 
All these operations (searching $e_p$ in $E$, computing its rank) are logarithmic in $|S|=|E|$.

\begin{examp}
  Figure~\ref{fig:threshlookup} shows six pairs $p_0,\ldots,p_5$, stock $S$, end time index $E$, and the skyband lower bound for $k=3$ (red line). For the pairs in $S$, we show similarity and end time (e.g., $(0.4, 5)$ for $p_2$); for the pairs in $E$, we only show the end time ($5$ for $p_2$). $S$ and $E$ are ordered by similarity resp.\ end time. We shift the orders by $k-1$ positions such that $E[v]$ is aligned with $S[k+v-1]$ (gray bars). Note that the pairs in the bars define the steps of the skyband lower bound, e.g., the first bar defines the point $(0.4,1)$, where the first step ends. This is a result of Theorem~\ref{lem:skybandpoints} and holds if the stock is minimal. We compute $lb(t,k)$ for $t=2.5$: $p_3$ at position $v=3$ is the smallest pair in $E$ with end time $\geq t$; the aligned pair $S[v+k-1]$ has similarity $0.2$, which is the skyband lower bound at time $t=2.5$.

  \begin{figure}[t]
   \centering
    \myasygraphics{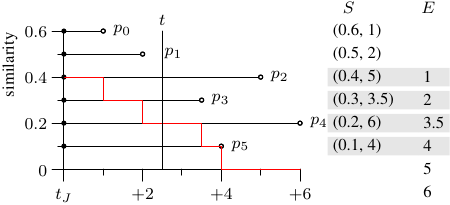}
   \caption{Threshold lookup at $t=+2.5$. $k=3$.}
   \label{fig:threshlookup}
  \end{figure}

\end{examp}

\subsection{Inserting New Pairs}

The insert operation adds a set of candidate pairs, $C$, to the stock. The challenge is to keep the stock minimal. New pairs  may turn out to be irrelevant (in which case they should not be inserted), or they may render other pairs irrelevant (which then must be removed). 

Assume we want to insert pair $p$ (dotted) into the stock in Figure~\ref{fig:irrelevant}. To check if $p$ is relevant, the rank at its end time $e_p$ must be at most $k$. The rank of $p$ is determined by the number of stock elements $p'$ that do not end before $p$ and are at least as similar, i.e., $e_{p} \leq e_{p'}$, $sim(p) \leq sim(p')$. There are $3$ such pairs ($p_2,p_3,p_5$, gray area); thus, $p$ is irrelevant (rank $4<k$ at end time). Note that inserting the irrelevant pair $p$ disrupts the alignment of $S$ and $E$ (gray horizontal bars) stated in  Theorem~\ref{lem:skybandpoints}.

\begin{figure}[htb]
 \centering
  \myasygraphics{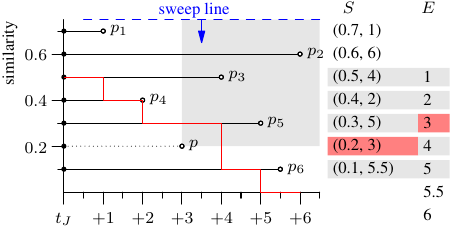}
 \caption{Relevant and irrelevant pairs, $k=3$.}
 \label{fig:irrelevant}
\end{figure}

\paragraph{Sweep Line Insertion} Let $p$ be the pair to be inserted. First, the relevance of $p$ must be checked. This is achieved using a sweep line algorithm that scans $S$ in sort order and counts all pairs $p'\in S$, $e_{p} \leq e_{p'}$, $sim(p) \leq sim(p')$ (gray area, Figure~\ref{fig:irrelevant}). If $p$ is irrelevant, it is rejected. Otherwise, $p$ is inserted, and all pairs $p''$, $e_{p} \geq e_{p''}$, $sim(p) \geq sim(p'')$ must be checked since they may have become irrelevant due to the insertion of $p$. For each pair $p''$, the  sweep line algorithm must be executed. Thus, the overall runtime is  $O(|S|^2)$.

\paragraph{Outline} We present our efficient insert algorithm in three steps. 
First, we present a \emph{cleanup} algorithm that uses end time index $E$ to remove all $i$ irrelevant pairs from stock $S$ in time $O(|S| + i \log |S|)$.  An insert algorithm that uses cleanup can add all candidates $C$ to the stock without any relevance checks and then remove all irrelevant pairs in one pass. This is a major improvement over the sweep line algorithm that is quadratic in $|S|$. 
Second, we optimize cleanup for the use with insert, where we know the candidate set $C$ up front.
Third, we present the efficient insert algorithm of SWOOP, which uses a merge approach and inserts pairs only if they are relevant. Intuitively, adding $C$ and cleaning the stock are interleaved.


\paragraph{Cleanup} The cleanup algorithm presented next removes all irrelevant pairs from stock $S$ for a given $k$. The algorithm uses the end time index $E$ and the following property of non-minimal stocks.

\begin{lem}
 \label{lem:endirr}
 If $u$ is the position of the first irrelevant pair in $E$, $p=E[u]$, then  the position of $p$ in $\stock$ exceeds $u+k-1$: $p=S[v], v>u+k-1$.
\end{lem}

\begin{proof} 
  By contradiction. Let $p=E[u]$ be the first irrelevant pair in $E$ and assume $p=S[v]$, $v\leq u+k-1$. The end time of all irrelevant pairs $p'$ is $e'\geq e_p$. Since $p=E[u]$, there are $u-1$ pairs that end before $p$.  None of these pairs can end at time $e_p$ since we order ties in $E$ by ascending similarity, i.e., irrelevant pairs precede relevant pairs. All $u-1$ pairs that end before $e_p$ must be more similar than any $p'$, otherwise $p'$ would render them irrelevant. Further, since $p$ is irrelevant, there must be at least $k$ additional pairs that are more similar than $p$ and are still valid at time $e_p$. Thus, in total at least $u+k-1$ pairs exist in $S$ that precede $p$.
\end{proof}

With Lemma~\ref{lem:endirr} we can clean the stock as follows: We scan $E$ and check for each position $u$ if the rank of $E[u]$ in $S$ exceeds $u+k-1$: in this case, the pair is irrelevant and is removed. We repeat the procedure from  position $u$ until all pairs in $E$ are processed. Computing the rank of $E[u]$ in $S$ has complexity $O(\log |S|)$.  We avoid the logarithmic factor in our cleanup algorithm (Algorithm~\ref{alg:optcleanupstock} without gray-shaded parts) as follows: We start with $e=1$ and iterate through the pairs $E[e]$ and $S[s]$ simultaneously such that $s=e+k-1$. If pair $E[e]$ sorts behind $S[s]$ in the sort order of $S$ then the rank of $E[e]$ in $S$ is above $e+k-1$, and $E[s]$ is irrelevant. Thus we avoid computing the exact rank of $E[e]$ in $S$. The overall complexity is $O(|S| + i \log |S|)$ for removing $i$ irrelevant pairs.

\begin{examp}
    We clean the stock in Figure~\ref{fig:irrelevant}, $k=3$. Initially, $e=1$ and $s=e+k-1=3$ (topmost gray bar). $E[1]=p_1$ does not sort after $S[3]=p_3$; thus, $p_1$ is relevant. Next step $e=2$: $E[2]=p_4$, $S[4]=p_4$, $p_4$ is relevant. For $e=3$, $E[3]=p$ sorts after $S[5]=p_5$; thus, $p$ is irrelevant and is removed. We proceed until $S$ is exhausted.
\end{examp}
%

\paragraph{Optimized Cleanup} Cleanup can be optimized for insertion by scanning only the regions of $S$ that may contain irrelevant pairs. We identify these regions by inspecting the set of inserted pairs, $C$.

\begin{theo}
    \label{lem:optimal-cleanup}
    Let stock $S$ be minimal, $C$ a candidate set of pairs, $maxs=max_{c\in C}(sim(c))$ and $maxe=max_{c\in C}(e_c)$ the maximum similarity resp.\  end time of all pairs $c\in C$. After adding $C$ to $S$ (without removing irrelevant pairs), the following holds for all pairs $p\in S$: if $p$ is irrelevant, then $sim(p)\leq maxs$ and $e_{p}\leq maxe$.
\end{theo}  

Optimized cleanup (Algorithm~\ref{alg:optcleanupstock} including gray-shaded parts) uses Theorem~\ref{lem:optimal-cleanup} to scan only those parts of $S$ and $E$ that might store irrelevant pairs. As an example, consider the stock in Figure~\ref{fig:irrelevant} and assume that the candidates $C=\{p_4,p_5\}$ have been inserted. With $maxs=sim(p_4)$ and $maxe=e_p$ we  only need to scan $p_4,p, p_5$. The algorithm starts the scan at $s=4$ in $S$ (since $p_4=S[4]$) and $e=s-k+1=2$ in $E$, and ends after three iterations.

\begin{algorithm}[t]
 \SetKwInOut{Inp}{Input}
 \SetKwInOut{Glob}{Globals}
 \SetKw{Break}{break}
 \SetKw{Return}{return}
 \SetKwProg{Fn}{Function}{}{end}
 \Glob{$\stock,E$: binary search trees (stock, end times), $k$.}
 \HiLi\Inp{$C$: candidates pairs.}
 \Fn{cleanup(\colorbox{gray!30}{C})}{
  \lIf{$|\stock| \le k$}{
   \Return
  }
  \HiLi$s\gets$ rank of $max_{c\in C}(sim(c))$ in $S$; $e\gets s - k +1$\;
  \HiLi\If{ $s<k$}{
   $e \gets 1$;  $s \gets k$\;
  }
  \While{$s\le |S|$ \colorbox{gray!30}{$\wedge E[e] \le max_{c\in C}(sim(c))$}}{
   \If{$E[e]>S[s]$ in sort order of $S$}{
     $s_e \leftarrow$ position of $E[e]$ in $\stock$\;
     remove $S[s_e]$ and $E[e]$;
   }
   \lElse{
    $e\gets e+1$;
    $s\gets s+1$
   }
  }
 }
 \caption{\colorbox{gray!30}{Optimized} cleanup.}
 \label{alg:optcleanupstock}
\end{algorithm}

\paragraph{Insert} The insert algorithm (cf. Algorithm~\ref{alg:insertstock}) processes both the stock items and the candidates in sort order of the stock (descending similarity), and a merge-like approach is used to verify candidate pairs before they are inserted. Intuitively, we walk along the skyband boundary (gray boxes in Figure~\ref{fig:toinsert}). Assume the current vertex of the skyband boundary is $v_i$. When we insert the candidates that fall between the vertexes $v_{i-1}$ and $v_{i}$, their end times must be above the end time $t_{bound}$, i.e., the end time of $v_{i-1}$. Irrelevant candidates are never inserted, but the insertion of relevant candidate pairs may render other pairs irrelevant. Since irrelevant pairs can only appear after the current position in $S$, they will be removed as we proceed (like in the cleanup algorithm).    

\begin{algorithm}[tb]
 \SetKwInOut{Inp}{Input}
 \SetKwInOut{Glob}{Globals}
 \SetKw{Break}{break}
 \SetKw{Return}{return}
 \SetKw{Continue}{continue}
 \SetKwProg{Fn}{Function}{}{end}
 \Glob{$\stock,E$: binary search trees (stock, end times), $k$.}
 \Inp{candidate pairs $C=(c_1,\ldots,c_{|C|})$ sorted by descending similarity.}
 \Fn{insert($C$)}{
  \tcp{Special case $|S|<k$}
  $i \leftarrow \min\{\max\{k-|\stock|,0\},|C|\} + 1$\nllabel{alg:stockinsert:upk0}\;
  \lIf{$|\stock| < k$}{\nllabel{alg:stockinsert:upk}
   insert $(c_1,\ldots,c_{i-1})$ into $\stock$ and $E$
  }
  \lIf{$|\stock|\le k \wedge i -1 = |C|$}{
   \Return\nllabel{alg:stockinsert:returnearly}
  }
  \tcp{Initialize $t_{bound}$ and indices $e$, $s$}
  $s\gets \max_{s} (sim(S[s]) > \max\{sim(c) \mid c\in C\})+1$\;
  \lIf{ $s\le k$\nllabel{alg:stockinsert:initindicesbegin}}{
   $e \gets 1$;  $s \gets k$; $t_{bound}\gets \indextime$
  }
  \lElse{
    $e\gets s - k +1$; $t_{bound}\gets E[e-1]$\nllabel{alg:stockinsert:initindicesend}
  }
  \tcp{Loop over $\stock$ and $E$}
  \While{$s \le |S| \wedge E[e] \le max_{c\in C}(sim(c))$\nllabel{alg:stockinsert:mainloopbegin}}{
   \tcp{Insert relevant candidates}
   \While{$i\le|C| \wedge sim(c_i)>sim(S[s]) $\nllabel{alg:stockinsert:innerloopbegin} }{
    \lIf{$e_{c_i}> t_{bound}$}{
     insert $c_i$ into $\stock$ and $E$
    }
     $i\leftarrow i+1$;\nllabel{alg:stockinsert:innerloopend}
   }
   \If{$E[e]>S[s]$ in sort order of $S$\nllabel{alg:stockinsert:cleanupbegin}}{
     $s_e \leftarrow$ position of $E[e]$ in $\stock$\;
     remove $S[s_e]$ and $E[e]$;
   }
   \lElse{
    $t_{bound}\gets E[e]$;
    $s\gets s+1$; $e\gets e+1$\nllabel{alg:stockinsert:mainloopend}
   }
  }
  \tcp{Insert remaining candidates}
  \While{$i\le|C|$\nllabel{alg:stockinsert:reminsertbegin}}{
   \If{$e_{c_i}>E[e-1]$}{
    insert $c_i$ into $S$ and $E$; $e\gets e+1$\;
   }
   $i\gets i+1$\nllabel{alg:stockinsert:reminsertend}\;
  }
 }
 \caption{Insert into stock.}
 \label{alg:insertstock}
\end{algorithm}

\begin{figure}[tb]
 \centering
  \myasygraphics{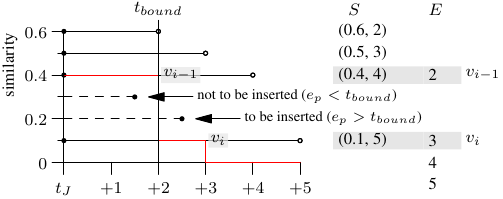}
 \caption{To insert or not to insert.}
 \label{fig:toinsert}
\end{figure}

Lines \ref{alg:stockinsert:upk0}--\ref{alg:stockinsert:returnearly} deal with the special case $|S|<k$. 
%
%
Lines \ref{alg:stockinsert:initindicesbegin}--\ref{alg:stockinsert:initindicesend} (similar to the cleanup algorithm) initialize end time threshold $t_{bound}$ and the positions $s$, $e$: $s$ is the rank of the first candidate in the sort order of $S$ (in a stock $S\cup\{c_1\}$); $e$ is aligned such that $(E[e], S[s])$ defines a skyband boundary vertex (gray bars in the figure). If the resulting $s$ is smaller than $k$, $s$ is initialized to $k$ and $e$ to 1. 

In the next step, the algorithm loops over $S$ and $E$ (lines \ref{alg:stockinsert:mainloopbegin}--\ref{alg:stockinsert:mainloopend}). In the inner loop, the relevant candidates that are more similar than $S[s]$ are inserted (lines \ref{alg:stockinsert:innerloopbegin}--\ref{alg:stockinsert:innerloopend}). Note that a candidate $c_i$ is inserted at position $s$, so $c_i$ becomes $S[s]$, and the loop exits after the first insertion (as $sim(c_{i+1})<sim(c_i)=sim(S[s])$). The relevance of a candidate is determined using the end time threshold $t_{bound}$ as illustrated in Figure~\ref{fig:toinsert}. The main loop proceeds like the cleanup algorithm (lines \ref{alg:stockinsert:cleanupbegin}--\ref{alg:stockinsert:mainloopend}), except that also $t_{bound}$ is updated.

After scanning the whole skyband boundary, there may still be candidates left  (lines \ref{alg:stockinsert:reminsertbegin}--\ref{alg:stockinsert:reminsertend}). This is the case for candidate pairs that are less similar than the least similar pair in $S$. Some of these pairs may be irrelevant. The end time for this check is the last end time in the skyband boundary, $E[e-1]$.

The complexity of insert depends on the sizes of $\stock$ and $C$. Inserting or deleting a pair takes $O(log |\stock|)$. Potentially each candidate pair has to be inserted, and each pair from $\stock$ has to be removed, yielding a worst-case complexity of $O((|\stock|+|C|)\log (|\stock|+|C|))$.

\section{Experiments}
\label{sec:experiments}

\subsection{Experimental Setting}

\paragraph{Setup} We conduct the experiments on an 8-core Intel Xeon E5-2630 v3 CPUs with 2.4~Ghz, 96 GB of RAM, and 20~MB cache (shared across cores), running Debian~9. Our code is written in C++ and is compiled with GCC using the -O3 option.

\paragraph{Algorithms}
We compare SWOOP with the following algorithms:

\begin{itemize}[noitemsep,topsep=0pt,parsep=0pt,partopsep=0pt]
\item SCase: State of the art for top-$k$ joins over streams~\cite{Shen2014}.
\item Static: State of the art for top-$k$ joins on static collections of sets~\cite{conf/icde/XiaoWLS09}; we adapt the algorithm to streams by reevaluating the top-$k$ join each time the sliding window changes.
\item Base: Baseline algorithm as presented in Section~\ref{sec:indbaseline}.
\item Static: Whenever a new set arrives, we run the top-$k$ set similarity join algorithm by Xiao et al.\cite{conf/icde/XiaoWLS09} to compute the top-$k$ from scratch.
\end{itemize}

 
 We implemented all algorithms  in C++\footnote{Source code will be published.}
 using data structures that are available from STL and
 Boost\footnote{\url{http://www.boost.org/}}.  For the binary search trees
 $\stock$ and $E$ in SWOOP, we use the Boost Multiindex container. We define one
 Multiindex structure that stores the stock $\stock$ and provide two indices
 (for $\stock$ and $E$) on this container.

\paragraph{Datasets} In our empirical evaluation, we use five data streams with different characteristics. Table~\ref{tab:datasets} shows the stream length (number of sets), the average set size, and the size of the token universe (number of distinct tokens) for each of the streams.

{TWEET.} Geocoded tweets collected
at Daisy\footnote{\url{http://www.daisy.aau.dk/}}
from February to April 2017. A tweet is a set of words with the posting time as a timestamp.

{DBLP.} Articles from DBLP\footnote{\url{http://dblp.uni-trier.de/}}~\cite{DBLP:journals/pvldb/Ley09}. A set is a publication and the tokens correspond to the words in the authors and title fields. The timestamp is the modification date from DBLP's XML file.

{FLICKR.} Photo meta-data\footnote{Provided by Bouros et al.~\cite{boge12}.}. A set consists of tokens from the tag or title text describing a photo. The timestamps are assigned randomly between 0 and 10,000 seconds. 

{ENRON.} E-mail data. A set is formed by the words in the subject and body fields, and the timestamp is defined by the send time.

{INDUSTRY.} Workflow instances from an ERP system. A set consists of pairs of subsequent workflow activities, and the timestamp is that of the last activity in the workflow.

\begin{table}[t]
\centering\small
\begin{tabular}{c|c|c|c}
Dataset & steam length & avg.\ set size & universe size\\\hline
TWEET & $3.4 \cdot 10^7$ & 13.44 & $3.7\cdot 10^7$\\\hline
FLICKR& $1.2 \cdot 10^6$ & 10.05 & $8.1\cdot 10^5$\\\hline
DBLP& $5.5\cdot 10^6$ & 12.10 & $1.7\cdot 10^6$\\\hline
ENRON& $2.5\cdot 10^5$ & 302.2 & $7.3\cdot 10^5$\\\hline
INDUSTRY& $4.9\cdot 10^7$ & 13.07 & $1.1\cdot 10^4$\\
\end{tabular}
\caption{Dataset statistics.}
\label{tab:datasets}
\end{table}

\paragraph{Measures}
The \emph{average window size} $\avgwin$ is the average number of sets in sliding window $W$, which is controlled by the duration $w$ of sliding window $W$.  

\emph{Pre-candidates} are the set pairs that must be formed when a new set arrives in the stream. In Base and SCase, a new set will form a pre-candidate with each set in the sliding window. In SWOOP and Static, the number of pre-candidates is the number of processed index list items.
\emph{Candidates} are the pre-candidates that are sent to the stock for insertion. Base sends all pre-candidates (with similarity larger than zero) to the stock. SWOOP and SCase filter the pre-candidates using a lower bound. Static does not use a stock and recomputes the join result for each window position.

The \emph{set rate} is the average number of processed sets per second and thus measures the performance of an algorithm. We map string tokens to integers as discussed in Section~\ref{sec:opttoken}; this process is identical for all algorithms and is not considered in the set rate.
The \emph{latency} is the time difference between the appearance of a set in the stream and the update to the top-$k$ result. It includes candidate generation, stock update, and potential waiting times in the input queue. 


\subsection{Scalability}

We evaluate the scalability of SWOOP and its competitors. We vary the
window size and the result size $k$, and we use all datasets.
Figure~\ref{fig:overview} shows the results. Missing values for an algorithm indicate that the stream could not be processed within 20k seconds (FLICKR, ENRON) resp.\ 200k seconds (other datasets).

\paragraph{Scalability in the window size} We measure the set rate for different window sizes $\avgwin$. For a small window size close to $k$, even Base performs well. For larger windows, however, the set rates of Base, SCase, and Static decrease sharply. When we increase the window size  by a factor of 10, the set rate of SCase decreases by a factor of 3.1 to 8.7, the set rate of Base by a factor of 15 to 76, the set rate of Static by a factor of up to 6.7. SWOOP clearly outperforms all other approaches and scales well with the window size. In fact, for $k=10$ the performance between $\avgwin=10^2$ and the largest window tested on the respective dataset decreased by less than a factor two; for a larger result size of $k=10^3$, we observe a similar behavior starting with $\avgwin=10^3$.

The DBLP stream is particularly challenging due to its skewed distribution of the timestamps. We show the results for varying window durations $w$ (the average window size is not meaningful for DBLP since it is heavily skewed). Base and Static run into a timeout even for the smallest window duration of $w=1$ day. SCase is slower than SWOOP by two to three orders of magnitude, and only SWOOP is capable of processing the DBLP stream for all window sizes  without timeouts. The set rate of SWOOP is affected little by the window size. 

\paragraph{Scalability in k} In Figure~\ref{fig:flickrres-scal-k}, we vary the result size $k$ for a fixed window size $\avgwin=10^3$ on the FLICKR stream, which all algorithms can process for $k=10$. The set rate of Base is low, but does not depend on $k$. This is because Base does not leverage lower $k$ values to decrease the stock size or reduce the number of candidates. SCase, Static, and SWOOP run faster for smaller $k$ values; SWOOP is consistently faster than SCase and Static by more than an order of magnitude.

\paragraph{Performance analysis} We analyze the performance advantage of SWOOP over its competitors in detail.

\emph{(1) Pre-candidates.} Figure~\ref{fig:tweetres-precand} shows the number of pre-candidates on the TWEET stream. Base and SCase form a pre-candidate with each set in the sliding window, which leads to a large number of pre-candidates. SWOOP uses the candidate index to reduce the number of pre-candidates that must be considered. The candidate index is highly effective: SWOOP considers only a small fraction of the pairs that its competitors must process, and the number of pre-candidates grows slowly with the window size. This explains SWOOP's scalability to large windows.

\emph{(2) Candidates.} In Figure~\ref{fig:tweetres-candidates} we measure the number of candidates. Base cannot prune any pre-candidates, and all pre-candidates are added to the stock. SCase and SWOOP both maintain the same pairs in the stock, so the number of candidates is the same. While SCase recomputes the stock from scratch for each new set in the stream, SWOOP updates the stock incrementally.

 \begin{figure*}[t!]
{\centering
  \myasygraphics[valign=t]{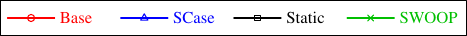}
  \subfigure[Set rate, INDUSTRY, $k=10$.]{
  \myasygraphics[valign=t]{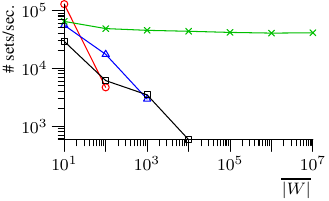}
  \label{fig:industryres-setspersecond}
  }%
  \subfigure[Set rate, ENRON, $k=10$.]{
  \myasygraphics[valign=t]{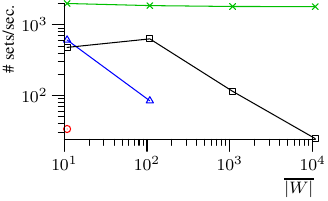}
  \label{fig:enronres-setspersecond}
  }%
  \subfigure[Set rate, DBLP, $k=10$.]{
  \myasygraphics[valign=t]{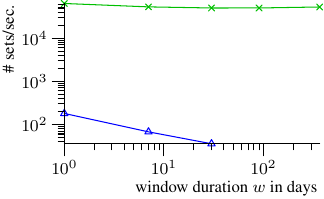}
  \label{fig:dblpres-setspersecond}
  }
  \subfigure[Set rate, FLICKR, $k=10$.]{
  \myasygraphics[valign=t]{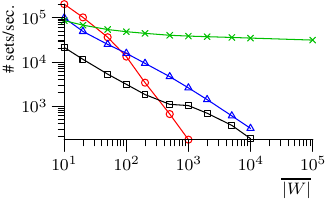}
  \label{fig:flickrres}
  }%
  \subfigure[Set rate, FLICKR, $k=10^3$.]{
  \myasygraphics[valign=t]{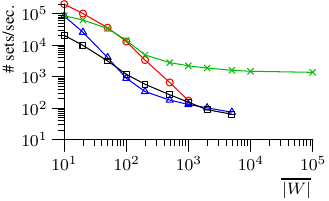}
  \label{fig:flickrres-large-k}
  }%
  \subfigure[Set rate, FLICKR, $\avgwin=10^3$.]{
  \myasygraphics[valign=t]{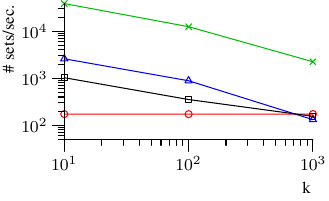}
  \label{fig:flickrres-scal-k}
  }
  \subfigure[Set rate, TWEET, $k=10$.]{
  \myasygraphics[valign=t]{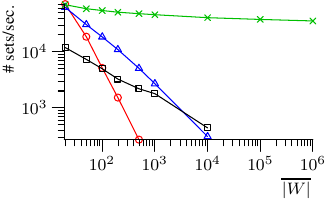}
  \label{fig:tweetres-setspersecond}
  }%
  \subfigure[Pre-candidates, TWEET, $k=10$.]{
  \myasygraphics[valign=t]{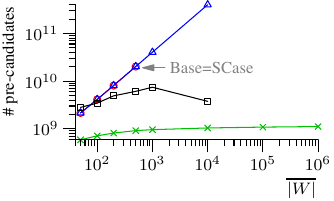}
  \label{fig:tweetres-precand}
  }%
  \subfigure[Candidates, TWEET, $k=10$.]{
  \myasygraphics[valign=t]{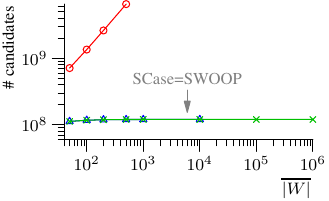}
  \label{fig:tweetres-candidates}
  }%
  \caption{Scalability: Set Rate, Pre-Candidates, Candidates.}%
  \label{fig:overview}
}%
 \end{figure*}

\emph{(3) Stock maintainance. } We evaluate the effect of the incremental stock maintenance vs.\ the candidate index in Figure~\ref{fig:hybrids}. To this end, we implement a version of SWOOP without a candidate index (labeled \emph{no-index}) and another version that recomputes the stock from scratch like SCase, i.e., it does not support incremental updates (labeled \emph{no-increment}). 

Clearly, both the candidate index and the incremental stock maintenance contribute to the performance of SWOOP. For large $k$, the bounds used by the candidate index are looser, which leads to more pre-candidates and reduced effectiveness (cf.\ Figure~\ref{fig:flickr-setspersecond-scal-k-hybrids}). The incremental index update, on the other hand, gains more for larger values of $k$ and outperforms the no-increment variant by up to an order of magnitude. When the window size grows (cf.\ Figure~\ref{fig:tweet-setspersecond-hybrids}), removing the candidate index leads to poor performance; the incremental index update outperforms no-increment, and the gain is almost independent of the window size.

Summarizing, the performance of SWOOP is mainly due to (a) the candidate index, which controls the number of pre-candidates as the window size grows, and (b) the incremental stock maintenance, which is up to an order of magnitude faster than recomputing the stock from scratch.

\begin{figure}[ht]
 \subfigure[Set rate, FLICKR, $\avgwin=10^3$.]{
  \myasygraphics{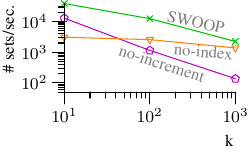}
  \label{fig:flickr-setspersecond-scal-k-hybrids}
 }%
 \subfigure[Set rate, TWEET,  $k=10$.]{
  \myasygraphics{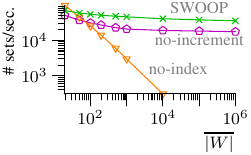}
  \label{fig:tweet-setspersecond-hybrids}
 }%
 \caption{SWOOP without a candidate index (no-index) and without incremental stock maintenance (no-increment).}
 \label{fig:hybrids}
\end{figure}

\emph{(4) Static algorithm.} Static does not maintain a stock. Instead, the join result is computed from scratch whenever the sliding window changes. This approach does not scale to large window sizes since the join time depends on the number of sets in the window. 

Note that Static cannot process new sets in batches: Each new set that enters the window may change all values of the top-$k$ result. Therefore, an approximation that processes batches of size $b>1$ ($b=1$ is the exact algorithm) may introduce a large error. The error rate, measured as the ratio between windows with the correct vs.\ windows with an incorrect top-$k$ results, is $O(1-1/b)$. The error is also high in practice. For example, the error is 65\% for batch size $b=100$ on ENRON ($|W|=1000$, $k=10$); more than 75\% of the incorrect top-$k$ lists differ by more than one element.

\subsection{Latency}  

To study the latency of SWOOP, we modify the timestamps in the TWEET dataset in order to produce a stream with a constant number of sets per second. We load SWOOP with 80\% of the average stream rate for the respective window size and measure the latency. The latencies are small: For $\avgwin=10^4$ ($4.72\cdot 10^4$ sets/second), the maximum latency is 0.25s with a maximum queue of 12,015 sets, and for  $\avgwin=10^6$ ($3.62\cdot 10^4$ sets/second), the maximum latency is 0.03s with a maximum queue length of 1365 sets. Interestingly, the latency is lower for larger windows. We attribute this effect to the skyband lower bound, which is looser for small windows (and fewer pairs in the stock). This may lead to more pre-candidates for individual sets. In fact, the maximum processing time  (candidate generation plus stock update) of a set is 0.04s for $\avgwin=10^6$ and 0.10s for $\avgwin=10^4$. This effect is limited to individual sets and does no show in the overall number of pre-candidates (cf.\ Figure~\ref{fig:tweetres-precand}).

\subsection{Optimized Token Processing Order}

We measure the effect of the processing order of the index lists during candidate generation in SWOOP.


In Section~\ref{sec:opttoken}, we propose to process the  index lists in ascending order of their length. We compare SWOOP, which uses this optimization, to SWOOP-noopt that uses the token order established based on the first appearance of a token.

We run the experiment on all datasets. For TWEET, FLICKR, and ENRON, we see almost no runtime difference, indicating that the token order is a good estimate of the real frequency in the stream. The picture is different for DBLP:
Figure \ref{fig:dblpres-setspersecond-noopt} shows that SWOOP can process the DBLP stream at a rate between 36 and 83 times faster than SWOOP-noopt.  The reason is the skew in the DBLP dataset. First, the sets are received in the stream at a very irregular rate, such that the window size $|W|$ varies between 0 and 338,199 for $w=1$~day (cf.\ Figure~\ref{fig:dblp-windowsize}). For large window sizes, the index lists grow long, and a poor list order has major effects on the performance. Second, the tokens 'Page' and 'Home' are only introduced at the positions 2,018 and 9,764, respectively.
However, these tokens become very frequent later (between 10\% and 50\% for most of the stream), as Figure~\ref{fig:dblp-tokenfrequency} shows (due to high correlation, the blue curve for 'Page' almost exactly tracks the red curve of 'Home'). As a result,  these tokens get assigned token numbers for infrequent tokens. Even worse, the largest frequency (almost 100\%) of these tokens occurs during the spikes in the window size, leading to very large numbers of pre-candidates  (cf.\ Figure~\ref{fig:dblpres-precands}). 

This offers empirical evidence that the optimization of the token order is relevant for difficult streams that are highly skewed.

\begin{figure*}[t!]
 \subfigure[Set rate, $k=10$.]{
  \myasygraphics{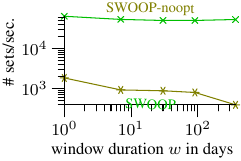}
  \label{fig:dblpres-setspersecond-noopt}
 }
 \subfigure[Pre-candidates, $k=10$.]{
  \myasygraphics{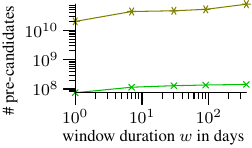}
  \label{fig:dblpres-precands}
 }
 \subfigure[Window size $|W|$ for window duration $w=1$ day.]{
  \myasygraphics{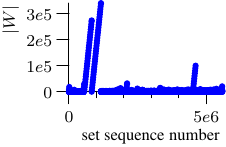}
  \label{fig:dblp-windowsize}
 }
 \subfigure[Frequency of terms ``Home'' and ``Page''.]{
  \myasygraphics{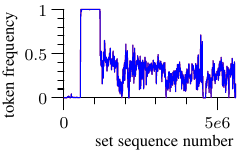}
  \label{fig:dblp-tokenfrequency}
 }
 \label{fig:dblp-all}
 \caption{Optimized Token Processing Order (DBLP).}
\end{figure*}


\subsection{Stock Size}
\label{sec:stocksize}

We study the maximum stock size for SWOOP, SCase, and Base. Specifically, we consider the maximum number of pairs that were stored in the stock during the processing of a particular stream. The stock size of Base is quadratic in the window size $\avgwin$, as it stores all pairs  (with non-zero overlap) in the window. The stocks of both SWOOP and SCase are minimal and of size $O(k\cdot |S|)$ in the worst case. 
Figure~\ref{fig:flickr-stocksize} shows the stock size for increasing window sizes $\avgwin$ and increasing values of $k$. As expected, the stock size of Base grows fast with the window size. Interestingly, the size of the minimal stock of SWOOP and SCase grows much slower than the worst case, indicated by the dotted lines. The stock size of Base is independent of $k$, as it stores all pairs (with non-zero overlap) --- see Figure~\ref{fig:flickrres-scal-k-stocksize}. The minimal stock of SWOOP and SCase is well below the worst case and also grows slowly: At $k=10$, the maximum stock size is $1.5\cdot10^2$, while at $k=1000$, it is $6\cdot10^3$, which is substantially below the worst case minimal stock size.
These results are in line with previous findings~\cite{Shen2014}, where the asymptotic behavior of the \emph{expected} stock size is shown to be $O(k\cdot \log(\avgwin/k))$.
Overall, the advantage of maintaining a minimal stock is clearly supported by our experiments.

\begin{figure}[ht]%
  \subfigure[$k=10$.]{%
  \myasygraphics{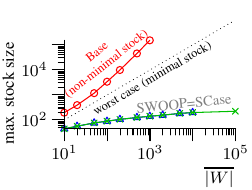}
  \label{fig:flickrres-stocksize}
  }%
  \subfigure[$\avgwin: 10^3$.]{%
  \myasygraphics{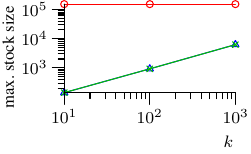}
  \label{fig:flickrres-scal-k-stocksize}
  }
  \caption{Maximum stock size in $k$ and $\avgwin$ (FLICKR).}
  \label{fig:flickr-stocksize}
 \end{figure}

\section{Related Work}
\label{sec:relwork}



Several proposals exist for the threshold-based set similarity joins on static data~\cite{deng-pvldb-2015,mann-tr-2015,wang-pvldb-2017}.  Deng et al.~\cite{deng-pvldb-2015} leverage  the pigeonhole principle on set partitions to prune candidates. A particularly successful concept is the so-called prefix filter~\cite{SCVGRK06}, which has been exploited in many set join algorithms~\cite{RBYMRS07,boge12,maau14,LARTH11,JWGLJF12,wang-pvldb-2017,CXWWXLJXYGW11}. Neither set partitioning nor prefix filtering can be applied in our top-$k$ settings as they require a fixed threshold. Wang et al.~\cite{wang-sigmod-2016_local-set-similarity} study a threshold-based similarity join on two windows that slide over a query and a document, respectively; a window defines a fixed-length set. In our setting, the sliding window covers all valid sets in the stream at a specific point in time.

Morales et al.~\cite{de2016streaming} consider sets that arrive in a stream. Their join computes all pairs of sets that are more similar than a user-defined threshold. They support an extended Cosine similarity measure that also considers the age of pairs using a pre-defined time-decay parameter. Their algorithm maintains all pairs that are more similar than the pre-defined threshold. This algorithm cannot be applied in our setting because (i) the time-decay cannot be modified to simulate a sliding window, and (ii) in order to enable top-$k$ functionality, the algorithm must support changing the threshold whenever a set enters or exits the window such that exactly $k$ pairs are maintained, which it does not.

Recent works propose top-$k$ search over static collections of sets~\cite{zhu-sigmod-2019_josie-topk,kocher-sigmod-2019_tasm-topk}, whereas we study the problem of top-$k$ joins over streams.
Xiao et al.~\cite{conf/icde/XiaoWLS09} consider the top-$k$ join scenario in a static setting where all sets are known up front. The processing is by token, not by set. The tokens are processed by decreasing positional upper bound. The algorithm is not applicable to our problem, unless we were to run the algorithm whenever window $W$ changes. We compare empirically with this approach.

Shen et al.~\cite{Shen2014} introduce SCase, a generic framework for computing the top-$k$ most similar pairs over sliding windows of object streams. The similarity function is supplied by the user, and no optimizations specific to sets are included. SCase uses four data structures for maintaining the stock: binary trees for the stock (i) sorted by similarity and (ii) sorted by end time, and (iii) for storing the skyband boundary; and (iv) a heap for the reconstruction of the three trees. We only need the first two data structures. We further require fewer operations and less memory, as we maintain the data structures incrementally rather than reconstructing them for each new set on a stream. We conduct a detailed empirical comparison with this approach.

A number of studies (e.g.,~\cite{Pripuzic2015,mouratidis2006continuous}) compute top-$k$ queries over streams of objects with a fixed score. All objects have the same lifetime, which is determined by a sliding window. In our setting, the lifetime of an object (pair of sets) is determined by the lifetime of two sets and varies between objects, which poses additional challenges compared to what is supported by these algorithms.


Ilyas et al.~\cite{DBLP:journals/vldb/IlyasAE04} compute top-$k$ join queries in relational databases. Tuples are joined on equality and are ranked based on the rank of the joined tuples. Furthermore, the algorithm requires static input. In our setting, the sets (tuples) have no rank associated with them. We compute the rank solely on the pairs of sets (joined tuples). Our join result changes based on the content of the sliding window. Therefore, this algorithm cannot be applied in our setting.

\section{Conclusions}
\label{sec:conclusion}
We presented a novel algorithm for continuous top-$k$ similarity joins over streams of sets. We introduced the notion of well-behaved similarity function to characterize the class of supported similarity functions. Our algorithm  integrates new set-based optimizations and a novel, incremental technique to maintain the join result. An extensive empirical comparison with the state-of-the-art algorithm SCase and a baseline offer evidence that the new algorithm is capable of outperforming its predecessors by up to three orders of magnitude.


\balance
\bibliography{mybib}
\bibliographystyle{abbrv}



\end{document}